%% file: ind-set.tex
\documentclass[]{amsart}

\usepackage{svg}

\usepackage{amsmath}
\usepackage{amssymb}
\usepackage{amsfonts}

\usepackage{amsthm}

\usepackage{MnSymbol}

\usepackage{thmtools}
\usepackage{thm-restate}

\usepackage{verbatim}

\usepackage[T1]{fontenc}
\usepackage[utf8]{inputenc}

\usepackage{enumitem}

\usepackage{mathtools}
\DeclarePairedDelimiter{\ceil}{\lceil}{\rceil}
\DeclarePairedDelimiter{\floor}{\lfloor}{\rfloor}

\usepackage{xcolor}

\usepackage{hyperref}
\hypersetup{
    colorlinks,
    linkcolor={black},
    citecolor={black},
    urlcolor={black}
}

\title{Properties of nowhere dense graph classes related to independent set problem}

\author{Grzegorz Fabiański}
\thanks{This work is supported by the NCN grant 2016/21/D/ST6/01485}

\newtheorem{theorem}{Theorem}
\newtheorem*{theorem*}{Theorem}

\newtheorem{lemma}{Lemma}
\newtheorem{proposition}{Proposition}
\newtheorem{corollary}{Corollary}
\newtheorem{remark}{Remark}
\newtheorem*{claim}{Claim}
\theoremstyle{definition}
\newtheorem{definition}{Definition}

\newcommand{\Rea}{\mathbb{R}}
\newcommand{\Nat}{\mathbb{N}}

\newcommand{\Cl}{\mathcal{C}}

\newcommand{\Al}{\mathcal{A}}

\newcommand{\Pl}{\mathcal{P}}

\newcommand{\so}{\rightarrow}

\newcommand{\nc}{\mbox{nc}}

\DeclareMathOperator{\dist}{dist}
\DeclareMathOperator{\profile}{profile}

\DeclareMathOperator{\set}{set}

\begin{document}

\begin{abstract}
  
  A set is called $r$-independent, if every two vertices of it are in distance greater then $r$. In the $r$-independent set problem with parameter $k$, we ask whether in a given graph $G$ there exists an $r$-independent set of size $k$. In this work we present an algorithm for this problem, which applied to a graph from any fixed nowhere dense class, works in time bounded by $f(k, r)|G|$, for some function $f$. We also present alternative algorithm, with running time bounded by $g(k, r)|G|$, working on slightly more general classes of graphs.
\end{abstract}
\maketitle

\tableofcontents

\section{Introduction}
\paragraph{Independent set.}  An $r$-independent set is a subset of vertices of a  given graph such that any two different vertices in the set have mutual distance greater than $r$. We focus on the following version of this problem: given a graph $G$ and integers $r, k$ we are asked to decide whether the given graph has an $r$-independent set of size $k$.

  The unparametrized version of $r$-independent set is \textit{NP}-hard, even in very restricted settings --- for example for $r=1$ and with restriction to planar graphs of maximum vertex degree $3$, as shown in \cite{DBLP:journals/tcs/GareyJS76}. Despite this, practically efficient algorithms are known. Examples include branching algorithms \cite{DBLP:journals/ieicet/TomitaSHW13} \cite{DBLP:conf/isaac/XiaoN13}, and algorithms utilizing reductions \cite{DBLP:journals/tcs/AkibaI16}. Cases of planar graphs and graphs of bounded degree are studied in detail, because branching factors can be lifted from restricted to general settings \cite{DBLP:journals/algorithmica/BourgeoisEPR12}.

   Another line of research is focused on the parametrized version of the independent set problem --- we measure the running time not only by size of input, but also in terms of $k$.  Independent set admits a trivial algorithm with running time $O(n^{f(k)})$ and it is believed that no quantitative improvement is possible. In particular, the hypothesis (called $\textit{W}[1]\not=$\textit{FPT}) that there is no algorithm working in time $f(k)n^c$ (where $c$ is a constant and $f$ is computable function) is widely accepted in parametrized complexity theory and \textit{FPT}-reduction to independent set problem serves as a basic tool for showing hardness. Parametrized problems solvable in time $f(k)n^c$, for some computable function $f$ and constant $c$ are called fix-parameter trackable algorithm (\textit{FPT}). In search for \textit{FPT}-algorithms, it is necessary to study this problem in more restricted settings, where both graph the class and parameter $k$ are restricted.
  
  The dual problem of $r$-dominating set is similarly hard and is well-studied on restricted graph classes. Our work inherits basic concepts from the corresponding study of $r$-dominating set on nowhere dense classes of graphs.
 
\paragraph{Nowhere dense graph classes.}

  Classes of nowhere dense graphs were introduced by Nešetřil and Ossona de Mendez in \cite{DBLP:journals/jsyml/NesetrilM10} and generalize many known families of sparse graphs, for example the following classes  are nowhere dense: graphs with bounded treewidth, planar graphs and classes defined by a forbidden (even topological) minor. Many problems are effectively solvable on general classes of nowhere dense graphs, allowing to replace complicated topological arguments, originating from Robertson–Seymour graph structure theorem, by more general and simpler combinatorial proofs.
\begin{definition}  
  We say that a graph $H$ is an $r$-shallow minor of a graph $G$, if there is a family $(B_h)_{h \in H}$ of connected subgraphs of $G$ (called minor model) such that the radius of each $B_v$ is at most $r$ and for each edge $uv$ in the graph $H$ there are $u' \in B_u, v' \in B_v$ such that $u'v'$ is an edge in $G$. We write $H \preceq_r G$ to denote that $H$ is a $r$-shallow minor of $G$.
\end{definition}
\begin{definition}
  We say that a class of graphs $\Cl$ is nowhere dense if for every $r$ there is $t$ such that $K_t \not\preceq_r G$ for all $G \in \Cl$ (here $K_t$ is a clique with $t$ vertices).
\end{definition}
  Nowhere denseness can be characterized by many other conditions (some of them are described in \cite{DBLP:books/daglib/0030491}). In the scope of this work, a characterization in terms of the splitter game will be used. Details will be provided in Section~\ref{Splitter}.

\paragraph{Dominating set.} A subset of vertices $X$ of a graph $G$ is called $r$-dominating if every vertex in $G$ is at distance at most $r$ from some vertex of $X$. The dominating set problem is the problem of determining if a given graph $G$ has an $r$-dominating set of a given size $k$. The dominating set problem is a natural dual to the independent set problem (dominating set is a covering problem, while independent set is a packing problem).
  
  In absence of restrictions on the class of graphs, parametrized $1$-independent set can be fpt-reduced to $1$-dominating set. This shows that on general graphs $1$-dominating set is at least as hard as independent set --- in fact it is believed that it is strictly harder (this hypothesis is called $W[1] \not= W[2]$).

  The computational complexity of dominating set on restricted graph classes is well studied. \textit{FPT}-algorithms were developed in bigger and bigger classes of graphs, including classes of graphs defined by excluded minors. This line of research culminates in showing that nowhere dense graph classes are the largest subgraph-closed classes on which $r$-dominating set admits \textit{FPT} algorithm. More precisely, in \cite{DBLP:conf/stacs/DrangeDFKLPPRVS16} the following dichotomy was shown:
\begin{theorem}
  Let $\Cl$ be a class of graphs, closed under taking subgraphs. Then:
\begin{itemize}
\item if $\Cl$ is nowhere dense, then for every $r$, the $r$-dominating set problem is in FPT.
\item if $\Cl$ is not nowhere dense, then there is $r$ such that the $r$-dominating set problem is as hard as on general graphs (it is $W[2]$-hard).
\end{itemize}
\end{theorem}

  The following notion of domination core is crucial for this work:
\begin{definition}
  Fix a graph $G$ and two subsets of vertices $X$ and $A$. We say that \emph{$X$ $r$-dominates $A$} whenever every vertex from $A$ is at distance at most $r$ from some vertex of $X$.

  By this terminology, being an $r$-dominating set is equivalent to $r$-dominating $V(G)$.
\end{definition}
\begin{definition}
  Fix a graph $G$. We say that a set of vertices $A$ is a $(k,r)$-domination core if for every set $X$ with $|X| \leqslant k$ (candidate for $r$-dominating set), if $X$ $r$-dominates $A$, then $X$ $r$-dominates $G$.
\end{definition}
  The notion of domination core relaxes requirements posed on the set in question by changing the global condition ($r$-domination of whole graph) to a local one (domination of core). The main technical result of [13] can be summarized as:
\begin{theorem} \label{SmallDomCore}
  For every nowhere dense class $\Cl$ of graphs and integers $k,r \in \Nat$ there is constant $C$ such that: for every graph $G \in \Cl$ there exists a domination core (with parameters $k,r$) in $G$ of size at most $C$. Moreover, it can computed in polynomial time.
\end{theorem}

\paragraph{Our results.}
  We prove the following:
\begin{theorem}
  Fix a nowhere dense class $\Cl$ of graphs. Then there is an algorithm deciding the existence of an $r$-independent set of size $k$ in given graph $G \in \Cl$ in time $f(k,r)|G|$, for some computable function $f$.
 Moreover, if such an an independent set exists then the algorithm in question will generate such a set.
\end{theorem}
  This result is a strengthening of a result from \cite{DBLP:conf/stoc/GroheKS14}, where an algorithm with running in time $f(k,r)|V(G)|^{1+\epsilon}$ (for every $\epsilon > 0$) is given.
  
  We use the following definition, analogous to the definition of a domination core.
\begin{definition}
  Fix a graph $G$. A \emph{$(k, r)$-witness} (that $G$ does not have an $r$-independent set of size $k$ in $A$) is a set $Q \subseteq G$ such that for every $X \subseteq V(G)$ with $ |X| = k$, there is a path, beginning and ending at a vertex in $X$, that passes through $Q$ and has length at most $r$.
  
  (Intuition: $X$ is not $r$-independent, and this fact is witnessed by $Q$.)
\end{definition}

  Our main technical result is the following:
\begin{restatable}[Witness existence]{theorem}{WitnessExistenceGen}
\label{WitnessExistence}
  Fix $k, r \in \Nat$ and a nowhere dense class of graphs $\Cl$. Then there is constant $C$ depending on $k, r$ and $\Cl$ such that for every graph $G \in \Cl$ which does not have an $r$-independent set of size $k$, there is a $(k,r)$-witness $Q \subseteq V(G)$ with $|Q|\leqslant C$.
\end{restatable}
 In our proof, a key role is played by the notion of a cowitness defined below. 
\begin{definition}  
  Fix a graph  $G$ with a set of vertices $A \subseteq V(G)$. A \emph{$(k, r)$-cowitness} for the pair $(A, G)$ is a set $Q \subseteq V(G)$ such that for every $X \subseteq A$ with  $|X| = k$ and $A \subseteq N_r(X)$, for every $a \in A$, there is a path beginning at $a$, ending at some vertex belonging to $X$, passing through $Q$, and of length at most $r$.
  (Intuition: the fact that $a \in N_r(X)$ is witnessed by $Q$ for every $X$ which covers all of $A$ --- not only $a$.)
\end{definition}

   As a side effect, in this work we found a new characterization of nowhere dense graph classes --- nowhere dense graph classes are exactly classes of graphs which admit a \emph{cowitness} of constant size. We define the notion of cowitness and state the characterization below:
\begin{restatable}[Characterization by small cowitnesses]{theorem}{cowitnessCharGen}
\label{Cowitness}
   Fix a class of graphs $\Cl$, closed under taking subgraphs. The the following conditions are equivalent:
   \begin{enumerate}
   \item $\Cl$ is nowhere dense,
   \item For all $r, k \in \Nat$ there is a constant $c$ such that for every $G \in \Cl$ and every $A \subseteq G$ there is a $(k, r)$-cowitness for the pair $(A, G)$ of size at most $c$,
   \item For all $r \in \Nat$ there is a constant $c$ such that for every $G \in \Cl$ there is a $(1, r)$-cowitness for the pair $(V(G), G)$ of size at most $c$.
   \end{enumerate}
\end{restatable}

\paragraph{My contribution.}
   In this work, two algorithms for the $r$-independent set problem are described. The ladder algorithm is invented by Michał Pilipczuk, Szymon Toruńczyk and Sebastian Siebertz. They originally developed an analogous algorithm for $r$-dominating set and adopted for independent set. The proof of correctness of the modified algorithm requires an analogous statement to Theorem \ref{SmallDomCore}, which is given by Theorem \ref{WitnessExistence}. The statement dualizes easily, but the known proof does not.

  My contribution is the proof of Theorem \ref{WitnessExistence}. During the proof I developed the definition of cowitness and found that it allows to characterize nowhere dense graph classes. The proof is constructive and allows a direct construction of an algorithm for $r$-independent set, which is the second algorithm presented here.
  
  I wish to thank Michał Pilipczuk and Szymon Toruńczyk, who guided me during this work.

\paragraph{Organization of this work.}
  In Chapter 2 we recall some standard definitions and theorems from the theory of nowhere dense graph classes and prove Theorems \ref{WitnessExistence} and \ref{Cowitness}. In Chapter \ref{chapterAlgo} we use this result to construct two linear-time algorithms for the $r$-independent set  problem --- the first is a direct consequence of the proof given, and the second is the original ladder algorithm.

\section{Witness existence}

\subsection{Notation}

 Fix parameters $r \in \Nat$ (radius of independence), $k \in \Nat$ (size of an $r$-independent set to be found) and a graph $G$.
 
  We write $A \subseteq G$ as a shortcut for $A$ being a subset of vertices of $G$. Similarly $u \in G$ means that $u$ is a vertex of $G$. The subgraph induced by a set of vertices $A$ will be denoted by $G[A]$. We measure the size of the graph $G$ by $|G|:=|V(G)|+|E(G)|$.

 The following definition implicitly depends on $r$ and $k$:
\begin{definition}
  \emph{A short path} is a path of length at most $r$.

  We say that \emph{a set $Q \subseteq G$ captures a set $Y \subseteq G$} if and only if there exists a short path connecting two different vertices in $Y$ that passes through some vertex in $Q$ (intuition: the set $Q$ witnesses that $Y$ is not an $r$-independent set).

  We say that \emph{$Q \subseteq G$ captures a pair $(Y, z)$}, where $z \in G$ and $Y \subseteq G$, if and only if there exists a short path connecting some $y \in Y$ with $z$ passing through some vertex in $Q$ (intuition: the set $Q$ witnesses that $\dist(Y, z) < r$).
\end{definition}
  We will use the following simple fact, which is immediate from the definition:
\begin{claim}
  Capturing is a monotone property: if $X \subseteq Y, Q \subseteq W$ and $X$ is captured by $Q$, then $Y$ is captured by $W$ as well.

  Version for pairs: if $X \subseteq Y$, $Q \subseteq W$, $a \in G$ and $(X, a)$ is captured by $Q$ then $(Y,a)$ is captured by $W$.
\end{claim}
  
  The notions of a witness and of a cowitness are key combinatorial properties in this study. We repeat the definition, now using the terminology of capturing:

\begin{restatable}{definition}{witnessDefGen}
  Fix a graph $G$. A \emph{$(k, r)$-witness} (that $G$ does not have an $r$-independent set of size $k$ in $A$) is a set $Q \subseteq G$ such that every $X \subseteq G$, with $ |X| = k$, is captured by $Q$ (intuition: $X$ is not $r$-independent, and this fact is witnessed by $Q$).
\end{restatable}

\begin{definition}  
  Fix a graph $G$ with a set of vertices $A \subseteq G$. A \emph{$(k, r)$-cowitness} for $(G, A)$ is a set $Q \subseteq G$ such that for every $X \subseteq A$ with  $|X| = k$ and $A \subseteq N_r(X)$, for every $a \in A$, $(X, a)$ is captured by $Q$ (intuition: the fact that $a \in N_r(X)$ is witnessed by $Q$ for every $X$ which covers all of $A$ --- not only $a$).
\end{definition}
  Being a witness also has the monotonicity property: if $Q$ is a witness with parameters $k, r$, then it is also a witness with parameters $k', r$ for any $k' > k$.
  
  The goal of this section is to prove (we repeat the statements from the introduction):
\WitnessExistenceGen*
  This theorem follows from Theorems \ref{Cowitness} and \ref{CowitnessIsWitness}, given below:
\cowitnessCharGen*
   We prove implication $(1) \so (2)$ in Section \ref{CowitnessExistence}. Implication $(2) \so (3)$ is immediate, while implication $(3) \so (1)$ will be proved in Section \ref{CowitnessToNoDense}.

   The notion of a cowitness is much less intuitive than the notion of a witness --- a cowitness does not witness any property of a graph (in any $G$, the whole vertex set is a cowitness for every $A \subseteq G$). Despite this, we have:

\begin{restatable}[Cowitness is a witness (whenever possible)]{theorem}{cowitnessIsWitnessR}
 \label{CowitnessIsWitness}
   Fix a graph $G$ and a $(r,k-1)$-cowitness $Q$ for the pair $(V(G), G)$. If in $G$ there is no $r$-independent set of size $k$, then $Q$ is a witness of this.
\end{restatable}

 We devote Section \ref{CowitnessIsWitnessProof} to the proof of this theorem.
 
\begin{proof}[Proof of Theorem \ref{WitnessExistence}]
 By Theorem \ref{CowitnessIsWitness}, it is enough to show the existence of cowitness of size bounded by $C$ (with modified parameter $k$). Such a cowitness exists by implication $(1) \so (2)$ of Theorem \ref{Cowitness}.
\end{proof}

\subsection{Preliminaries}
 In this section, we provide some definitions and results used in the proof of Theorem \ref{WitnessExistence} and \ref{CowitnessIsWitness}.

 Through this section we fix $r \in \Nat$.
\subsection{Definitions}
\begin{definition}
  By $\dist^H(v,w)$ we will denote distance between vertices $v$ and $w$ computed in graph a $H$, ie. the length of a shortest path in $H$, connecting $v$ and $w$.
  
  The $r$-neighborhood of a vertex  $x$ in a graph $g$ is the set $\{v \in G | \dist^H(v, x) \leqslant r\}$, denoted $N_r^H(x)$.
  We drop the superscript, writing $\dist(v,w)$ and $N_r(v)$, when $H$ is clear from context.
\end{definition}

  We will be interested in distances from some vertex $v$ to a (small) set $S$, but only up to some threshold. We call such functions distance profiles. Formally, we define:
\begin{definition}
  A \emph{profile on $S$} is a function $S \so [0, 1, \ldots, r, r+1]$ ($r+1$ should be seen as equivalent to $+\infty$). A set of profiles on $S$ will be dentoes as $\Pl(S)$.
  
  The \emph{profile associated to a vertex} $v \in G$ with respect to $S \subseteq G$ is the profile on $S$ such that for $s \in S$, $\profile_S(v)(s) = \min( \dist(v, s), r+1)$.
\end{definition}
  The following straightforward claim shows the importance of profiles:
\begin{claim}
  Fix $X, Q \subseteq G$. For every $a \in G$, whether $Q$ captures $(X, a)$ depends only on $\profile_Q(a)$. More precisly, if $a, b \in G$ are such that $\profile_Q(a)=\profile_Q(b)$ then $Q$ captures $(X, a)$ if and only if $Q$ captures $(X, b)$.
\end{claim}
  We need a slightly stronger property --- we have to understand the dependence on the set $X$ as well as on $a$. This motivates the following notions:

\begin{definition}
   Fix $S \subseteq G$. The \emph{profile associated to a subset $X \subseteq G$} is the following profile on $S$:
  \[  \profile_S(X)(s) :=  \min_{x \in X} \min(\dist(s, x), r+1) \]

  The \emph{trace} of a function $p : S \so \Nat$ in a graph $G$ is the set defined as
  \[ [p] := \bigcup_{s \in S} N^G_{p(s)}(s) \]
\end{definition}  
  We can now state the stronger proposition:
\begin{proposition}\label{ProfileLemma}
 Fix $Q \subseteq G$. Then for every $X \subseteq G$ and $a \in G$, $Q$ captures $(X, a)$ if and only if the function $\profile_Q(a) + \profile_Q(X)$ does not have all values greater than $r$.
 
 The set of such $a$ is equal to $[r-\profile_Q(X)]$ (where $r$ is considered as a constant function).
\end{proposition}
  The proof is immediate from definitions. Note that if $p$ is a profile then $r-p$ has range $\{-1, 0, \cdots, r\}$ and that $N_{-1}(v) = \emptyset$.  

\paragraph{Nowhere denseness.}
\label{Splitter}
  We present a characterization of nowhere dense graph classes obtained in [1]. The following statements are slight reformulations of Theorem 4.2 and Remark 4.3 from that work.

  \begin{definition} \emph{(Splitter game)} Let $G$ be a graph and let $l, r \in \Nat$. The $(l, r)$-splitter game on $G$ is
played by two players, \texttt{Connector} and \texttt{Splitter}, as follows. We let $G_0 := G$. In round $i + 1$ of the
game, Connector chooses a vertex $v_{i+1} \in G_i$ (further game play will be restricted to $N_r(v_{i+1})$). Then \texttt{Splitter} picks a vertex $w_{i+1} \in N_r^{G_i}(v_{i+1})$ --- a vertex to be removed from arena. We let new arena be $G_{i+1} := G_i[N_r^{G_i}(v_{i+1} ) - \{w_{i+1}\}]$. \texttt{Splitter} wins if $G_{i+1} = \emptyset$. Otherwise the game continues
on $G_{i+1}$. If \texttt{Splitter} has not won after $l$ rounds, then \texttt{Connector} wins.
 \end{definition}
 
\begin{theorem} \label{SplitterGame} Let $\Cl$ be a class of graphs. Then $\Cl$ is nowhere dense if and only if for every $r \in \Nat^+$ there is $l \in \Nat^+$, such that for every $G \in \Cl$, Splitter wins the $(l, r)$-splitter game on $G$.
\end{theorem}

\begin{theorem}
   There is an algorithm $\Al$ which given a graph $G$ and $r \in \Nat$ plays a splitter game (interactively, playing as splitter) such that: for every nowhere dense class $\Cl$ of graphs, there is $d \in \Nat$ such that $\Al$ wins in $d$ rounds.
   
   Moreover, algorithm $\Al$ works in linear time --- given a move of the opponent, a response in the game is generated in linear time.
\end{theorem}

  We will use bounds on neighborhood complexity in nowhere dense graph classes. In \cite{DBLP:conf/icalp/EickmeyerGKKPRS17} the following result is proved.
\begin{theorem} \label{neigComp}
  Fix a class $\Cl$ of nowhere dense graphs. Then there is a function $\nc_{\Cl}$, such that for every $\epsilon \in \Rea, r \in \Nat$, every $G \in \Cl$ and every $A \subseteq G$, the number of different profiles on set $A$ obtained by vertices from $G$ is bounded by $\nc_{\Cl}(r, \epsilon)|A|^{1+\epsilon}$. 
\end{theorem}

\paragraph{Greedy lemma.}
  We will use the following lemma: 
\begin{lemma}[Greedy lemma] \label{GreedyLemma}
  Fix a graph $G$, a set $X \subseteq G$ and $r, k \in \Nat$. Then either:
\begin{enumerate}[label={(\arabic*)}]
 \item There is $Y \subseteq X$ such that $|Y| > k$ and $Y$ is $2r$-independent in $G$ (equivalently: for $x \neq y$ and $x,y \in Y$, balls $N_r(x), N_r(y)$ are disjoint).
 \item There is $Z \subseteq X$ such that $|Z| \leqslant k$ and $X \subseteq N_{2r}^{G}(Z)$.
\end{enumerate}
 Moreover, there is an algorithm which for given $G, X, r, k$ computes $Y$ or $Z$ as above in time $O(k|G|)$. 
\end{lemma}
\begin{proof}
 The lemma can be proved using a greedy algorithm, which aims at constructing the set $Y$. Let us assume we have already constructed $Y \subseteq X, |Y| \leqslant k$, which is $2r$-independent. There are two possibilities:
\begin{itemize}
\item If $X \subseteq N_{2r}(Y)$, then we find $Z:=Y$ satisfying the second condition and terminate.
\item Otherwise, there is $v \in X$ such that $v \not\in N_{2r}(Y)$, so $Y \cup \{v\} $ is $2r$-independent and larger than $Y$, hence we continue with $Y := Y \cup \{v\}$.
\end{itemize}
  After $k+1$ repetitions of this procedure, we constructed a set $Y$ satisfying the first condition. Testing whether there exists $v \not\in N_{2r}(Y)$ can be easily done in linear time, so the final running time is $O(k|G|)$.
\end{proof}

\subsection{Proof of cowitness existence}
\label{CowitnessExistence}
  This section is devoted to the proof of the implication $(1) \so (2)$ from Theorem \ref{Cowitness}. By characterization of nowhere denseness in terms of the splitter game (Theorem \ref{SplitterGame}), it is sufficient to show the following theorem.
  \begin{theorem}[Cowitness existence]\label{CowitnessExistence2}
  For every $d, k, r \in \Nat$ there is $c \in \Nat$ such that for every graph $G$ such that the $(d, 3r)$-splitter game on $G$ is won by splitter, and every $A \subseteq G$, there is a $(k, r)$-cowitness $Q$ for the pair $(A, G)$ of size at most $c$.
  \end{theorem}
  We want to prove this theorem by induction on $d$. To do so, we must strengthen the claim. We introduce a separator $S$ --- a small set of vertices, such that $G - S$ is simpler (in the sense that splitter can win in a smaller number of rounds). Additionally, throughout the recursion we will maintain the additional requirement that $S \subseteq Q$.
  \begin{theorem}[Cowitness existence --- extended version] \label{CowitnessExistenceExt}
  For every $d, k, r, s \in \Nat$ there is a $c \in \Nat$ such that for every graph $G$ and $S \subseteq G$ with $|S| = s$ such that the $(d, 3r)$-splitter game on $G' := G - S$ is won by splitter, and for every $A \subseteq G$, there is a $(k, r)$-cowitness $Q$ for the pair $(A, G)$ of size at most $c$, satisfying $S \subseteq Q$.
  \end{theorem}
  Obviously, Theorem \ref{CowitnessExistence2} follows from Theorem \ref{CowitnessExistenceExt} by considering $s:=0$. The remainder of this section is devoted to the proof of Theorem \ref{CowitnessExistenceExt}.
\bigskip

  Fix $d, r, s, G, A$. We want to find a $(k, r)$-cowitness $Q$ for the pair $(A, G)$.

  We proceed by induction on $d$. The initial step is obvious: $d = 0$ means that $V(G)=S$ and we can take $Q:=V(G)$ and $c=s$.
  
  
  Now we consider the induction step.

\begin{lemma} \label{zeroLemma}
  Take any $v \in G$. Define $w_v$ to be a wining splitter response for the connector play at $v$ in the $(d,3r)$-splitter game and $G_v$ to be $G[N_{3r}^{G'}(v) \cup S]$. Then $(G_v, S \cup \{w_z\}, A^*)$ satisfies the induction hypothesis, for every $A^* \subseteq G_v$.
\end{lemma}
\begin{proof}
 By definition of the splitter game, splitter win the $(d-1, 3r)$ splitter game in $G[N_{3r}(v) - \{w_v\}] = G_v - (S \cup \{w_v\})$.
\end{proof}

\begin{lemma} \label{localLemma}
  There is a set $Z$ of size bounded by a function of $k, r, s$, such that:
  For every pair $(X,a)$ with $|X|\leqslant k$, $A \subseteq N_r^G(X)$ and $a \in A$: if $(X, a)$ is not captured by $S$, then $a \in N^{G'}_{2r}(Z)$
\end{lemma}
\begin{proof}
  Fix $G, S, A, r, k$ and $k$ as in the statement. For a profile $p$ on $S$, define
 \[ T_p := \{ x \in A | \profile_S(x) = p \} \] 
  We apply Lemma \ref{GreedyLemma} in the graph $G'$ to $T_p$. If the second case of the lemma holds, say that the profile $p$ is covered, and define $Z_p$ to be the outcome of the application of the lemma. Thus $T_p \subseteq N_{2r}^{G'}(Z_p)$.

\begin{claim}
  If for some vertex $a \in  T_p$ and set $X \subseteq G$ with $|X|\leqslant k$, such that the pair $(X, a)$ is not captured by $S$, then the profile $p$ is covered. \end{claim}
\begin{proof}
  We show that we are always in case $2$ of the lemma. Assume otherwise: there is $Y_p$ such that $|Y_p|>k$ and balls $N_r^{G'}(y)$, for all $y \in Y_p$, are disjoint.
  
  By assumption $|X| \leqslant k$, so for at least one $y \in Y_p$, the ball $N_r^{G'}(y)$ is disjoint with $X$.
  
  We claim that $\dist^G(y, X ) > r$. We already know that this fact holds in $G'$. The only remaining thing to show is the impossibility of connecting $y$ with $X$ by a short path passing through $S$. But this is impossible by the assumption that $S$ does not capture $(X, a)$, because $y$ and $a$ have the same profiles on $S$.
\renewcommand\qedsymbol{$\righthalfcup$}
\end{proof}

  Define $Z$ as $ Z := \bigcup_{p} Z_p$, where sum ranges over all profiles $p$ which are covered.

   We claim that $X$ satisfies the required condition. Indeed, for every $a \in G$ and $X \subseteq G$ with $|X|\leqslant k$ such that $(X, a)$ is not captured by $S$, the $\profile^G_S(a)$ is covered and so $a \in T_p \subseteq N^{G'}_{2r}(Z_r) \subseteq N^{G'}_{2r}(Z)$. 
\end{proof}

\begin{lemma} \label{mainLemma}
 For each $z \in Z$ and profile $p$ on $S$ there is a set $A_{z,p}$ such that the following holds.
 
 For every $(X, a)$ which is not captured by $S$ there is a pair $(z, l)$ such that any cowitness $Q_{z,l}$ for $(G_z, A_{z,l})$ captures $(X, a)$.
\end{lemma}

\begin{proof}[Proof of Lemma \ref{mainLemma} $\so$ Theorem \ref{CowitnessExistenceExt}]
  Define $Q = \bigcup_{z \in Z,p \in \Pl}Q_{z,p}$, where $Q_{z,p}$ is a cowitnesses for the pair $(G_z, A_{z,p})$, obtained by applying the induction hypothesis to the $(G_z, S \cup \{w_z\}, A_{z,p})$.

  Then $|Q|$ is bounded by induction hypothesis.
\end{proof}

  It remains to prove the lemma.

\begin{proof}[Proof of lemma \ref{mainLemma}]
  For $z \in Z$. Define $A_z := A \cap N^{G'}_{2r}(z)$ and let $A_{z,p} := A_z - [r-p]$, where $p$ is a profile on $S$.

  Let $(X, a)$ ($X \subseteq G, a \in A$) be a pair which is not captured by $S$.

  By Lemma \ref{localLemma}, we can find $z \in Z$ such that $a \in N^{G'}(z)$. Let $l := \profile_S(X)$. We show that $z,l$ satisfies the required property.

  Let $X_z = X \cap  V(G_z)$. Let $Q_{z, l}$ be a cowitness for $(G_z, A_{z,l})$.

  We show that $Q_{z,l}$ captures $(X_z, a)$, which implies that $Q_{z,l}$ captures $(X, a)$ since $X_z \subseteq X$.

  As $Q_{z,l}$ is a cowitness for $(G_z, A_{z,l})$, it is enough to show the following properties:
\begin{enumerate}
\item $A_{z,l} \subseteq N_r^{G'}(X_z)$
\item $a \in A_{z,l}$
\end{enumerate}
  To prove these properties we define two subsets of $A_z$:
\begin{itemize}
 \item $A_z^1$ is the set of these vertices in $A_z$ that can be connected to a vertex from $X$ by a short path passing through $S$,
 \item $A_z^2$ is the set of these vertices in $A_z$ that can be connected to a vertex from $X$ by a short path laying inside $G'$.
\end{itemize}
 By $A \subseteq N_r^G(X)$, any vertex from $A_z$ satisfies at least one of these two conditions. The fact that $(X, a)$ is not captured by $S$ can be equivalently expressed as $a \not\in A_z^1$.

 By Proposition \ref{ProfileLemma}, $A_{z,l}= A_z - A^1_z$. As $a \not\in A_z^1$ and $a \in A_z$ then $a \in A_{z,l}$, proving $(2)$.

\begin{claim}
  $A_z^2 \subseteq N^{G_z}_r(X_z)$
\end{claim}
\begin{proof}
  Take $v \in A_z^2$ and fix a short path lying in $G'$ that connects vertex from $A$ with vertex $x \in X$. Since $A_z \subseteq N_{2r}^{G'}(z)$ and $\dist^{G'}(v, x) \leqslant r$ then $x \in N_{3r}^{G'}(z)$. So $x \in X_z$.
\renewcommand\qedsymbol{$\righthalfcup$}
\end{proof}
  In particular, $A_{z,l} = A_z - A_z^1 \subseteq A_z^2 \subseteq N^{G'}_r(X_z)$ proving $(1)$.

  By $(1)$ and $(2)$ and the fact that $Q_{z,p}$ is a cowitness for $(G_z, A_{z,l})$, we have that $Q_{z,l}$ captures $(X_z, a)$, hence also $(X, a)$.
\renewcommand\qedsymbol{$\blacksquare$}
\end{proof}

\begin{figure}[htbp]
  \centering
  \def\svgwidth{\columnwidth}
    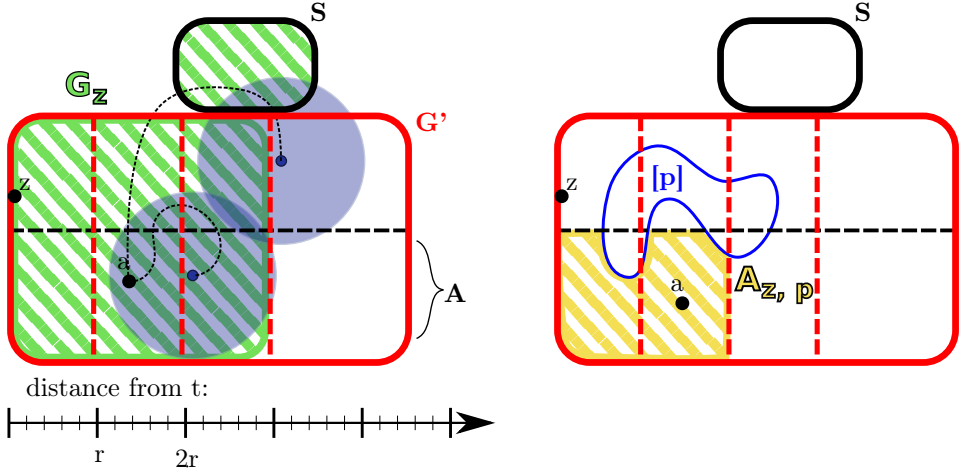
  \caption{Left: $G_t$, two vertices from $X$ with $r$-neighborhoods (taken in $G'$) and two short paths. Note that a short path passing through $S$ connects two vertices being apart in $G'$. Right: visualization of a profile and $A_{z,p}$}
\end{figure}

\begin{corollary} \label{CowitnessComp}
  In Theorem \ref{CowitnessExistence} we can bound size of the cowitness $Q$ by $|Q| \leqslant (rk)^{d(1+o(1))}$ for some constant $d$ depending only on class $\Cl$ and $r$. Moreover, a cowitness of that size can be computed in time $O(|G||Q|)$.
\end{corollary}
\begin{proof}
  Let $c(r,k,d,s)$ be the smallest number for which Theorem \ref{CowitnessExistenceExt} holds.
  
  From the proof given above we have that:
\[ c(r,k,d,s) \leqslant c(r,k,d-1,s+1)|Z|(r+1)^s = (k+1)(r+1)^{2s} c(r,k,d-1,s+1); \]
\[ c(r,k,0,s) \leqslant s. \]
  This gives upper bound:
\[ c(r,k,d,s) \leqslant (k+1)^{d}((r+1)^{2})^{s+(s+1)+\cdots (s+d-1) }c(r,k,0,s+d) = (d+s)(k+1)^d(r+1)^{2sd+d(d-1)} \]
 Which implies that:
\[ c(r,k,d,0) \leqslant d(k+1)^d(r+1)^{d(d-1)} \] 
  The cowitness can be computed: the provided proof is not only constructive, but effectively computable by a recursive algorithm:
\begin{itemize}
\item Compute $Z$ using computability of greedy lemma --- this phase is done in time $O(k|G|)$
\item Compute pairs $(G_z, A_{z,p})$ --- each such pair can be found in linear time
\item For each each pair above we make a recursive subcall.
\end{itemize}
  All steps are done in time (not counting execution of recursive subcalls) $|G|(R+k)$ where $R$ is a number of recursive subcalls. Always $k \leqslant R$, so running time is bounded by $|Q||G|$ (the fraction is the number of leaves of recursion). 
\end{proof}

 \subsection{Cowitness is already a witness}
 \label{CowitnessIsWitnessProof}
  
 We repeat statement for convenience:
\cowitnessIsWitnessR*
\begin{proof}
 Assume otherwise --- that $Q$ is a cowitness, but not a witness. Then there is $X \subseteq G$ with $|X|=k$ such that $Q$ does not captures $X$. From all such choices of $X$, fix one which minimizes the function:
  \[ f(X) = \left| \{ w \in X | \exists v \in X: v \not= w \land \dist(v,w) \leqslant r \} \right|. \]
 
 There is no a $r$-independent set of size $k$ in $G$, so there are $w,v \in X$ such that $w \not= v$ and $\dist(w,v) \leqslant r$.
 
 If $N_r(X - \{w\}) = V(G)$ then, by applying the definition of cowitness to the pair $(X - \{w\}, w)$, we get that $Q$ captures $X$, a contradiction.
 
 Assume otherwise, that $N_r(X - \{w\}) \subsetneq V(G)$. Take any $w' \not\in N_r(X - \{w\})$. Define $X' := X \cup \{w'\} - \{w\}$. We claim that $X'$ is a smaller counterexample:
 \begin{claim} $X'$ is not captured by $Q$.
 \end{claim}
 \begin{proof} 
 Assume it is: then there is a short path between $ x_1,x_2 \in X', x_1 \neq x_2$ going through $Q$. As $w' \not\in N_r(X' - \{w\})$, we have $x_1, x_2 \neq w$. The considered path connects two vertices from $X$, so $Q$ captures $X$, a contradiction.
\renewcommand\qedsymbol{$\righthalfcup$}
\end{proof}
 
 It only remains to show that $f(X')<f(X)$ --- by construction $f(X') = f(X - \{w\}) < f(X)$. This finishes the proof.
 \end{proof}

  The computational version of this problem can be solved efficiently:
\begin{theorem} \label{CowitnessIsWitnessComp}
  Given is graph $G$, parameters $k,r$, set $Q$ promised to be $(k-1, r)$-cowitness for pair $(V(G), G)$, and a set $X$ promised not to be captured by $Q$. Then an $r$-independent set of size $k$ in $G$ can be computed in time $O(k^2|G|)$.
\end{theorem}
\begin{proof}
  Such algorithm can be easily extracted from proof of Theorem \ref{CowitnessIsWitness} --- it describes a constructive procedure which from any $X$ not captured by $Q$ and not being an $r$-independent set (this is equivalent to $f(X) > 0$) generates $X'$ such that $|X'|=|X|=k, f(X') < f(X)$ and $Q$ does not capture $X'$. This procedure computes $X'$ in time $O(k|G|)$.
  
  Iterating it at most $k$ times we get $X'$ such that $f(X')=0$. This means that $X'$ is an $r$-independent set.
\end{proof}

 \subsection{Existence of a small cowitness implies nowhere denseness}
 \label{CowitnessToNoDense}
 We now prove implication $(3) \so (1)$ from Theorem \ref{Cowitness}:
\begin{theorem}
 Fix a class $\Cl$ of graphs closed under taking subgraphs. Assume that for every $r$ there is a constant $c$ such that for every $G \in \Cl$ there is a $(1, r)$-cowitness for the pair $(V(G), G)$ of size at most $c$. Then $\Cl$ is nowhere dense.
\end{theorem}
 
 We use following characterization of somewhere denseness (origins of this statement are not clear, see Claim 6.1 with discussion from \cite{DBLP:conf/stacs/DrangeDFKLPPRVS16}):
\begin{definition}
  A $r$-subdivision of a graph $G$ is a graph constructed by replacing edges from $G$ by paths of length $r+1$ (for every edge in $G$, new $r1$ vertices are added in the middle of the edge).
\end{definition}
\begin{theorem}
  Let $\Cl$ be closed under taking subgraphs. Then $\Cl$ is somewhere dense if and only if there is $r$ such that $\Cl$ contains the $r$-subdivision of every simple graph.
\end{theorem}

\begin{proof}[Proof of Theorem \ref{CowitnessToNoDense}]
  It is sufficient to prove that the $r$-subdivision of clique $K_n$ does not have $(1, r+1+\ceil{r/2})$-cowitness for pair $(V(G), G)$ of size smaller than $\floor{(n-1)/2}$. Assume there is such a cowitness $Q$.
 
  Define $D$ to be the set of all vertices of degree $n-1$, then $|D|=n$. By $|Q|\leqslant \floor{(n-1)/2}$ we know that there is $v \in D$ such that $D \not\in Q$. Take $X := \{v\}$. Then $N_{r+1+\ceil{r/2}}(X) = V(G)$ and $r+1+\ceil{r/2}$ is the smallest such constant.
 
  Take $a, b \in D$ to be any two different vertices, different from $v$. Take $c$ to be one of two vertices in the middle of the subdivided edge between $a$ and $b$ (in case when $r$ is even we choose the one that is nearer to $a$). We consider paths of length $r+1+\ceil{r/2}$ connecting $v$ and $c$ --- at least one of such paths must has nonempty intersection with $Q$. When $r$ is even, then there is exactly one such path, if $r$ is odd --- exactly two. Anyway, the union of all such paths is contained in the set $W_{a, b}$ consisting of subdivided edges $(v, a), (v, b), (a, b)$.

  By $v \not\in Q$ we have that $W_{a,b} - \{v\}$ have nonempty intersection with $Q$.
 
  The sets $W_{a,b} - \{v\}$ and $W_{c,d} - \{v\}$ are disjoint whenever the sets $\{a, b\}$ and $\{c, d\}$ are disjoint. By ranging over all possible indices (from $D - \{v\}$) we can construct $\floor{(n-1)/2}$ such disjoint sets. All such sets contain element from $Q$, which contradicts that $|Q|<\floor{(n-1)/2}$.
\end{proof} 

\section{Algorithmic aspects} \label{chapterAlgo}
  In this section we apply the obtained results to construct two algorithms for parametrized $r$-independent set in nowhere dense graph classes. The first approach uses direct cowitness search, as established in Theorem \ref{CowitnessComp}. The second one originates in stability theory and depends only on existence of small witnesses (Theorem \ref{WitnessExistence}).

\subsection{Algorithm for \texorpdfstring{$r$}{}-independent set by direct cowitness search}
\label{kernelization}
  We want to prove the following Theorems.


  

\begin{theorem} \label{WitnessChecking}
  The problem of checking, given a graph $G$ and subset $Q \subseteq G$, whether $Q$ is a witness for $G$ can be done in time $O(|G||Q|+2^{k|Q|(1+o(1))})$ on general graphs and (for every $\epsilon)$ in time $O(|G||Q|+\nc_{\Cl}(\epsilon, r)^k |Q|^{k(1+o(1))}$ on a nowhere dense graph class $\Cl$.
\end{theorem}

\begin{theorem} \label{AlgoInd}
  On nowhere dense classes $\Cl$, the problem of finding an $r$-independent set of size $k$ (or conclude that it does not exist) can be solved in time $O( |G|(rk)^{d(1+o(1))}+\nc_\Cl(\epsilon, r)^k (kr)^{dk(1+\epsilon+o(1))})$ for some constant $d$ which depends only on the class $\Cl$.
\end{theorem}

\begin{proof}[Proof of Theorem \ref{WitnessChecking}]
  We check if every $X \subseteq G$ with $|X|=k$ satisfies that there are distinct vertices $x_1, x_2 \in X$ and $q \in Q$ such that $\dist(x_1, q) + \dist(q, x_2) \leqslant r$.

  Observe that whether given $X$ satisfies this property depends only on the multiset of profiles $\{\{ \profile^Q_r(x) | x \in X\}\}$. So instead of considering all $X \subseteq G$ with $|X|=k$ it is sufficient to consider all multisubsets of size $k$ of $B := \{\{ \profile_r^Q(v) | v \in G\}\}$ --- the multiset of all realized profiles to $Q$ in $G$.

  The set $B$ can be computed in time $|Q||G|$ by running $BFS$ from every vertex from $Q$.

  Despite $|B|=|G|$, the number of different elements in $B$ is bounded by $(r+2)^{|Q|}$ --- the number of profiles to the set $Q$. In case when $G$ belongs to some nowhere dense class of graphs, for every $\epsilon >0$ we have the better bound by $\nc_\Cl(\epsilon,r)|Q|^{1+\epsilon}$ (Theorem \ref{neigComp}).

  It follows that the number of different submultisets of $B$ of size $k$ is bounded by $r^{|Q|k(1+o(1))}$ in case of general graphs, and by $\nc_\Cl(\epsilon,r)^k|Q|^{k(1+\epsilon)}$ in case of graphs belonging to a nowhere dense class $\Cl$.

 Checking each such submultiset is done in time $O(k|Q|)$ (for each $q \in Q$ we search for two nearest vertices from $X$).
\end{proof}

\begin{proof}[Proof of Theorem \ref{AlgoInd}]
  We can compute a cowitness $Q$ of size $(kr)^{d(1+o(1))}$ in time $|Q||G|$. Then, by Theorem \ref{CowitnessIsWitness}, the problem of deciding whether there is an $r$-independent set of size $k$ is equivalent to deciding whether $Q$ is a witness. By Theorem \ref{CowitnessIsWitnessComp}, this reduction also holds for the computational version of the independent set problem.
  
  By applying the algorithm from Theorem \ref{WitnessChecking} we are done.
\end{proof}

\subsection{Ladder algorithm}
  In this section we describe an alternative algorithm to find an $r$-independent set. It works correctly on any graph. When the input graphs are restricted to some class of graphs which satisfies additional properties (that are more general then the assumption of nowhere denseness), the presented algorithm works in linear \textit{FPT} time, otherwise there is no specific time bound.

  We will work with formulas of first order logic. For tuples $\overline{x}=(x^1, \ldots, x^n)$ and $\overline{y} = (y^1, \ldots, y^m)$, fix a formula $\phi(\overline{x}, \overline{y})$ of first order logic with free variables $x^1, \ldots, x^n$ and $y^1, \ldots, y^m$.
  
\begin{definition}
 A \emph{ladder} for formula $\phi(\overline{x}, \overline{y})$ of length $d$ is a family of tuples $\overline{x}_i, \overline{y}_j$ (of length respectively $n$ and $m$), defined for odd $i \in \{1, \ldots, d\}$ and even $j \in \{2, \ldots, d\}$, such that $\phi(\overline{x}_i, \overline{y}_j)$ holds if and only if $i<j$.
\end{definition}
\begin{definition}
 A class of graphs $\mathcal{C}$ is called \emph{stable} if and only if for any formula $\phi(\overline{x}, \overline{y})$ there is a constant $c$ such that there is no ladder for $\phi$ of length $c$ in any $G \in \mathcal{C}$.
\end{definition}
  Adler and Adler proved in \cite{DBLP:journals/ejc/AdlerA14} that:
\begin{theorem}\label{StabT}
  Fix a class $\Cl$ of graphs closed under taking subgraphs. Then $\Cl$ is stable if and only if $\Cl$ is nowhere dense.
\end{theorem}

\begin{definition}
 Let $\set(\overline{x})$ be the set of elements appearing in the tuple $\overline{x}$.

  Define $\psi_{r,k;n,m}(\overline{x}, \overline{y})$ to be the formula describing the property that $\set(\overline{x})$ is not captured by $\set(\overline{y})$. This property can be easily expressed in first order logic.
\end{definition}
\begin{remark}
  The formula $\psi_{r,k,n,m}$ is independent of the order of elements inside the tuple. Therefore, in the ladder for $\psi$ we can use sets instead of tuples if we additionally require that the sets $X_i, Y_j$ have sizes at most $n$ and $m$, respectively.
  
  By generalized ladder we will mean ladder of sets not necessarily satisfying that sets $X_i, Y_j$ have sizes bounded by $n$ and $m$. 
\end{remark}

  The goal of this section is to prove the following:
\begin{theorem} \label{ladder}
 There is algorithm $A$ that given a graph $G$, and $k, r \in \Nat$ produces either:
 \begin{enumerate}[label=(\alph*)]
 \item an $r$-independent set of size $k$, or
 \item a witness that $G$ has no $r$-independent set of size $k$
 \end{enumerate}
  On any class of graphs $\Cl$ satisfying both of the following assumptions, the algorithm $A$ works in time $f(k,r)|G|$ and generates witnesses of size bounded by a function of $r$ and $k$.
\begin{enumerate}
\item There is a constant $c$ such that for any graph $G \in \Cl$ if $G$ does not have an $r$-independent set of size $k$, then $G$ has a witness of this of size at most $c$
\item There is a constant $d$ such that there is no ladder for $\psi_{r,k;k,c}$ of length $d$ in any graph from $\Cl$. 
\end{enumerate}
\end{theorem}
\begin{remark} 
  Assumptions $1$ and $2$ are satisfied by every nowhere dense graph class, by Theorems \ref{Cowitness} and \ref{StabT}. Note that we do not require any computability of a witness in our class $\Cl$, only its existence.
\end{remark}
\begin{proof} [Proof of theorem \ref{ladder}]

  Our algorithm tries to construct a generalized ladder for $\psi_{r,k,k,c}$ step by step. When it fails to enlarge the ladder then we get either $(a)$ or $(b)$. The ladder in the construction will have the additional property that $|X_i| = k$ and, in presence of $(2)$, $|Y_i| \leqslant c$.
  
  Assume we have already constructed a ladder of length $d$. We have two steps:
 
\paragraph{$Y$-step: witness search.}
  Assume that $d$ is even. We want to find $Y_{d+1}$.
  
  If any $X_i$ constructed so far is $r$-independent, we can terminate concluding $(a)$. Assume otherwise. We aim at finding a set that witnesses that $X_i$ is not $r$-independent for all odd $i \in \{1, \ldots, d\}$. Assumption $(2)$, if present, guarantees the existence of such witness of size at most $C$.
  
  To find the witness we consider profiles. Whether $Q$ is a witness for all $X_i$ depends only on the set $\{ \profile_W(x) | x \in Q \}$, where $W = \bigcup_i X_i$. The size of $W$ is bounded by $kd$ --- so the number of profiles and number of sets of profiles are also bounded.
  
  We proceed as follows:
\begin{enumerate}
\item  We compute $\profile_W$ for all vertices in $G$. This can be done by starting \texttt{BFS} from each vertex in $W$ --- this procedure takes time $O(|W||G|)$. Denote the set of all found profiles by $P$.
\item  We iterate over all subsets of $P$ and check whether the corresponding set witnesses all $X_i$. As $Y_{i+1}$ we take the smallest witness.
\end{enumerate}
  The size of set $P$ is bounded by $r^{kd}$, thus the second step takes time $O(kd2^{|P|})=O(2^{(r^{kd})}kd)$. By choosing the smallest possible witness, we make the algorithm uniform: when $(2)$ holds, we do not need to know $c$ in advance, as the found witness will be of size at most $c$.

\paragraph{$X$-step: finding candidate for an $r$-independent set.}
  Assume that $d$ is odd. We want to find a set $X_{i+1}$ such that none of $Y_i$'s witnesses that $X_{i+1}$ is not an $r$-independent set.
  
  We check if $A_i=\bigcup_{j} Y_j$ is a witness --- by Theorem \ref{WitnessChecking} this can be done in time $f(|A_i|)|G|$ and in constructive way (if $A_i$ is not a witness, then a set of size $k$ that is not witnessed is generated). Now we can conclude:
\begin{itemize}
\item If $A_i$ is a witness then we can conclude $(b)$ --- note that the size of $A_i$ is bounded by $dC$.
\item If $A_i$ is not a witness, take $X_{i+1}$ to be any set not witnessed by $A_i$  --- it will not be witnessed by any $Y_i$ and meets the stated requirement.
\end{itemize}
 
\paragraph{Halting property.}
  On classes of graphs satisfying (1) and (2), we make at most $d$ rounds by assumption, because at every step of the algorithm, sets $X_1, Y_2, X_3, \ldots, $ form a ladder. The algorithm is correct by construction and we already argued that we produce a witness of size at most $dc$.
  
  To prove the halting property on general graphs we need to bound the length of the ladder. Observe that $A_l$ is a strictly increasing family of sets (as a function of $l$). After at most $n$ steps, $A_l = G$ ($l$ is the number of steps up to that point). We prove that we terminate in at most two steps.
  
  In the next step, the algorithm will be checking whether there exists set not captured by $A_l=V(G)$ --- this is equivalent to checking whether $G$ has an $r$-independent set. If such a set is not found then the witness $V(G)$ is returned. If such a set $B$ is found, then the ladder is enlarged by one step and the $r$-independent set $B$ is returned in the subsequent step.
\end{proof}

\section{Further work}
  For the ladder algorithm we use the bound of ladder length obtained by general argument of stability. For the particular $\psi$ that we used, maybe better bounds can be obtained in restricted settings, for example on planar graphs or graphs with bounded treewidth.
  
  Practical evaluation of this algorithm is of independent interest.  Various heuristic improvements seem reasonable --- including playing splitter game more efficiently, using partial cowitnesses found in another branch of the cowitness finding algorithm, identifying and pruning unimportant vertices found during the ladder construction.

\bibliographystyle{plain}
\addcontentsline{toc}{chapter}{Bibliography}

\bibliography{GroheKS14,GareyJS76,TomitaSHW13,AkibaI16,BourgeoisEPR12,XiaoN13,DrangeDFKLPPRVS16,AdlerA14,EickmeyerGKKPRS17,NesetrilM10,0030491}


\end{document}

%% file: drawing_exp.pdf_tex
\begingroup%
  \makeatletter%
  \providecommand\color[2][]{%
    \errmessage{(Inkscape) Color is used for the text in Inkscape, but the package 'color.sty' is not loaded}%
    \renewcommand\color[2][]{}%
  }%
  \providecommand\transparent[1]{%
    \errmessage{(Inkscape) Transparency is used (non-zero) for the text in Inkscape, but the package 'transparent.sty' is not loaded}%
    \renewcommand\transparent[1]{}%
  }%
  \providecommand\rotatebox[2]{#2}%
  \newcommand*\fsize{\dimexpr\f@size pt\relax}%
  \newcommand*\lineheight[1]{\fontsize{\fsize}{#1\fsize}\selectfont}%
  \ifx\svgwidth\undefined%
    \setlength{\unitlength}{367.27543113bp}%
    \ifx\svgscale\undefined%
      \relax%
    \else%
      \setlength{\unitlength}{\unitlength * \real{\svgscale}}%
    \fi%
  \else%
    \setlength{\unitlength}{\svgwidth}%
  \fi%
  \global\let\svgwidth\undefined%
  \global\let\svgscale\undefined%
  \makeatother%
  \begin{picture}(1,0.49819115)%
    \lineheight{1}%
    \setlength\tabcolsep{0pt}%
    \put(0,0){\includegraphics[width=\unitlength,page=1]{drawing_exp.pdf}}%
    \put(0.67628448,0.29715751){\color[rgb]{0,0,1}\transparent{0.99000001}\makebox(0,0)[lt]{\lineheight{0}\smash{\begin{tabular}[t]{l}\textbf{[p]}\end{tabular}}}}%
    \put(0,0){\includegraphics[width=\unitlength,page=2]{drawing_exp.pdf}}%
    \put(0.45878916,0.17409828){\color[rgb]{0,0,0}\makebox(0,0)[lt]{\lineheight{0}\smash{\begin{tabular}[t]{l}\textbf{A}\end{tabular}}}}%
    \put(0.42883889,0.35434441){\color[rgb]{1,0,0}\makebox(0,0)[lt]{\lineheight{0}\smash{\begin{tabular}[t]{l}\textbf{G'}\end{tabular}}}}%
    \put(0.31825646,0.46990005){\color[rgb]{0,0,0}\makebox(0,0)[lt]{\lineheight{0}\smash{\begin{tabular}[t]{l}\textbf{S}\end{tabular}}}}%
    \put(0.01258422,0.29210765){\color[rgb]{0,0,0}\makebox(0,0)[lt]{\lineheight{0}\smash{\begin{tabular}[t]{l}z\end{tabular}}}}%
    \put(0.01985681,0.07542185){\color[rgb]{0,0,0}\makebox(0,0)[lt]{\lineheight{0}\smash{\begin{tabular}[t]{l}distance from t:\end{tabular}}}}%
    \put(0.17671207,0.00005445){\color[rgb]{0,0,0}\makebox(0,0)[lt]{\lineheight{0}\smash{\begin{tabular}[t]{l}2r\end{tabular}}}}%
    \put(0.37375218,0.04307381){\color[rgb]{0,0,0}\makebox(0,0)[lt]{\lineheight{0}\smash{\begin{tabular}[t]{l} \end{tabular}}}}%
    \put(0.0917622,0.00549996){\color[rgb]{0,0,0}\makebox(0,0)[lt]{\lineheight{0}\smash{\begin{tabular}[t]{l}r\end{tabular}}}}%
    \put(0,0){\includegraphics[width=\unitlength,page=3]{drawing_exp.pdf}}%
    \put(0.11560878,0.20887253){\color[rgb]{0,0,0}\makebox(0,0)[lt]{\lineheight{0}\smash{\begin{tabular}[t]{l}a\end{tabular}}}}%
    \put(0,0){\includegraphics[width=\unitlength,page=4]{drawing_exp.pdf}}%
    \put(0.89044018,0.46985711){\color[rgb]{0,0,0}\makebox(0,0)[lt]{\lineheight{0}\smash{\begin{tabular}[t]{l}\textbf{S}\end{tabular}}}}%
    \put(0.58727943,0.29210765){\color[rgb]{0,0,0}\makebox(0,0)[lt]{\lineheight{0}\smash{\begin{tabular}[t]{l}z\end{tabular}}}}%
    \put(0,0){\includegraphics[width=\unitlength,page=5]{drawing_exp.pdf}}%
    \put(0.69686776,0.18589931){\color[rgb]{0,0,0}\makebox(0,0)[lt]{\lineheight{0}\smash{\begin{tabular}[t]{l}a\end{tabular}}}}%
  \end{picture}%
\endgroup%